\documentclass[11pt]{article}
\usepackage{fullpage,amsmath,amssymb,amsthm}

\newcommand{\R}{\mathbb{R}}
\newcommand{\dist}{\mathrm{dist}}
\renewcommand{\P}{\mathbb{P}}

\newcommand{\mc}[1]{\mathcal{#1}}
\newcommand{\len}{\mathrm{length}}

\def\to{{\rightarrow}}

\newtheorem{definition}{Definition}[section]

\newtheorem{theorem}[definition]{Theorem}

\newtheorem{fact}[definition]{Fact}

\newtheorem{corollary}[definition]{Corollary}
\newtheorem{proposition}[definition]{Proposition}
\theoremstyle{definition}

\newtheorem{example}[definition]{Example}

\newcommand{\schur}[2]{\texttt{Schur} \left(#1,  #2\right) }
\title{Localization of Electrical Flows}

\author{Aaron Schild\footnote{Supported by NSF grant CCF-1535977.} \\EECS, UC
Berkeley \\ \texttt{aschild@berkeley.edu} \and Satish Rao\footnote{Supported by NSF grant CCF-1535977.} \\ EECS, UC Berkeley\\ \texttt{satishr@cs.berkeley.edu}  \and Nikhil Srivastava\footnote{Supported by NSF grant
CCF-1553751 and a Sloan research fellowship.} \\Mathematics, UC
Berkeley\\ \texttt{nikhil@math.berkeley.edu}}

\begin{document}
\maketitle

\begin{abstract}
We show that in any graph, the average length of a flow path in an electrical flow between the endpoints of a random edge is $O(\log^2 n)$. This is a consequence of a more general result which shows that the spectral norm of the entrywise absolute value of the transfer impedance matrix of a graph is $O(\log^2 n)$. This result implies a simple oblivious routing scheme based on electrical flows in the case of transitive graphs.
\end{abstract}

\newcommand{\hP}{\overline{\Pi}}
\newcommand{\reff}{\mathrm{Reff}}
\section{Introduction}
Electrical Flows have have played an important role in several recent advances
in graph algorithms --- for instance, in the context of exact and approximate
maximum flow/minimum cut \cite{CKMST,LRS,madry13}, multicommodity flow
\cite{kelner2012faster}, oblivious routing
\cite{harsha08,lawler,kelner2011electric}, graph
sparsification \cite{spielman2011graph}, and random spanning tree generation
\cite{kelner2009faster, MST}. This is due to the
emergence of nearly linear time Laplacian solvers for computing them, beginning
with the work of Spielman and Teng \cite{spielmanteng}, and also to their well-known connections
with random walks. Using them to solve combinatorial problems is not typically
immediate, and may be likened to putting a square peg into a round hole: at a
high level, many of the traditional problems of computer science are concerned
with finding flows in graphs with controlled $\ell_1$ and $\ell_\infty$ norms
(corresponding to distance and congestion, respectively), whereas electrical
flows minimize the $\ell_2$ norm (energy). Reducing one to the other often
requires some sort of iterative method for combining many electrical flows with
varying demands and graphs.

In this work, we ask the following basic structural question about electrical
flows in arbitrary unweighted graphs: 
\begin{quote}
What is the typical $\ell_1$ norm of the
unit current electrical flow between two neighboring vertices $u,v$ in a graph?
\end{quote}
Recall that the $\ell_1$ norm of a unit circulation-free flow is the average
distance traveled by the flow, since any such flow $f_{uv}:E\rightarrow \R$ may
be decomposed as a convex combination of paths which all have the same direction of flow on every edge:
$$f_{uv} = \sum_{i \in \mathcal{P}_{uv}} \alpha_i f_i,$$
where $\mc{P}_{uv}$ is the set of simple paths from $u$ to $v$, and we have
$$\|f_{uv}\|_1 = \sum_{i\in \mc{P}_{uv}} \alpha_i\|f_i\|_1 = \sum_{i\in
\mc{P}_{uv}} \alpha_i\len(f_i).$$ Thus, this question asks when/whether electrical
flows in a graph travel a greater distance than shortest paths, and by how much.

\subsection{Three Examples}
To get a feel for this problem, and to set the context for our result and its
proof, we begin by presenting three instructive examples. We
will use the notation $b_v$ for the indicator vector of a vertex, $b_{uv}=b_u-b_v$ for the signed incidence vector of
an edge $uv$, $B$ for the $m\times n$ signed edge-vertex incidence
matrix of a graph (where the edges are oriented arbitrarily), and $L=B^TB$ for
the Laplacian matrix.  For any pair of vertices $u,v$, we will use the notation
$\Delta(u,v):=\|f_{uv}\|_1/\dist(u,v)$, where $f_{uv}$ is the unit electrical
flow between $u$ and $v$ and $\dist$ is the shortest path distance in the graph.

The first example shows that in general $\Delta(u,v)$ can be quite large for the
worst-case edge in a graph.
\begin{example}[Parallel Paths] Consider the graph  consisting of a single edge between vertices
$u$ and $v$ and $\sqrt{m}$ disjoint parallel paths of length $\sqrt{m}$ with
endpoints $u$ and $v$. Since the effective resistance of the parallel paths
taken together is $1$, half of the unit flow between $u$ and $v$ will use the
paths, assigning a flow of $1/2\sqrt{m}$ to each path, and the other half will
traverse the ege $uv$. Thus, we have $\Delta(u,v)=(\sqrt{m}+1)/2$. 
However, notice that for most of the other edges in the graph, $\Delta$ is tiny.
For instance, for any edge $ab$ near the middle of one of the parallel paths,
a $1-O(1/\sqrt{m})$ fraction of the flow will traverse the single edge, so we will
have $\Delta(a,b)=O(1)$. 
\end{example}

On the other hand, $\Delta(u,v)$ is uniformly bounded for every edge in an
expander.
\begin{example}[Expander Graphs] Let $G$ be a constant degree $d-$regular expander graph with
transition matrix $P$ satisfying $\|P-J/n\|\le \lambda$ for some constant
$\lambda$, where $J$ is the all
ones matrix. Letting $Q:=P-J/n$ and $E=I-J/n$, we have the power series
expansion orthogonal to the all ones vector:
$$(L/d)^+ = (E-Q)^+ = E+\sum_{k\ge 1}Q^k.$$ Now for every edge $uv$ we calculate
the electrical flow across its endpoints:

$$\|BL^+b_{uv}\|_1 \le \|B\|_{1\to 1}\|L^+\|_{1\to 1}
\|b_{uv}\|_1,$$
where $\|\cdot\|_{1\to 1}$ is the $1\to 1$ operator norm, i.e., maximum column
sum of the matrix. Let $T=O(\log n)$ be the mixing time of $G$, after which
$\|P^T-J/n\|_{1\to 1}=\|Q^T\|_{1\to 1}\le 1/n$. Noting that $\|B\|_{1\to 1}\le d$
and $Q^k = P^k-J/n$ and applying the triangle inequality, we obtain:
$$ \|BL^+b_{uv}\|_1 \le \frac{2d}{d}\sum_{k=0}^T (\|P^k\|_{1\to
1}+\|J/n\|_{1\to 1}) + \|L^+\|_{1\to 1}\cdot\|Q^T\|_{1\to 1}.$$
Since $P^k$ is a doubly stochastic matrix we have $\|P^k\|_{1\to 1}=1$ for all
$k$. Moreover, $\|L^+\|_{1\to 1}\le \sqrt{n}\|L^+\|\le \sqrt{n}/\lambda$.
Combining these facts, we get a final bound of $\Delta(u,v)=O(\log n)$, for every edge $uv\in
G$. 
\end{example}
We remark that bounds similar to (and somewhat more general than) the above were shown in
the papers \cite{lawler, kelner2011electric} using different techniques.

Finally, we note that there are highly non-expanding graphs for which
$\Delta(u,v)$ is also uniformly bounded, which means that expansion does not
tell the whole story.
\begin{example}[The Grid] Let $G$ be the $n\times n$ two dimensional torus (i.e., grid
with sides identified). Then it is shown in \cite{lawler} that for every
edge $uv\in G$ we have $\Delta(u,v)=O(\log n)$, even though $G$ is clearly not
an expander. We briefly sketch an argument explaining where this bound comes
from. Let $uv$ be any horizontal edge in $G$, and let $w$ be a vertex in $G$ at
vertical distance $k$ from $u$ and $v$. We will show that the potential at $w$
in the unit current $uv$-electrical flow is small, in particular that 
$$\phi(w):=b_w^TL^+b_{uv}=O(1/k^2).$$ First we recall (see,
e.g., \cite{bollobas2013modern} Chapter IX) that the potential at a
vertex $w$ when $u,v$ are fixed to potentials $-1,1$ is:
$2(\P_w(t_v<t_u)-\P_w(t_u<t_v))$, where $\P_w$ is the law of the random walk
started at $w$ and $t_u$ is the first time at which the walk hits $u$.
By Ohm's law, this means that: 
$$\phi(w)\le \left|\frac{2}{\reff(u,v)}(\P_w(t_v<t_u)-\P_w(t_u<t_v))\right|,$$
where $\reff(u,v):=b_{uv}^TL^+b_{uv}$ is the effective resistance of the edge $uv$.
Since the resistance of every edge in the grid is equal to $1/2$, we find that
$|\phi(w)|=O(|\P_w(t_v<t_w)-\P_w(t_w<t_v)|)$. 

We now analyze these probabilities. Roughly speaking, the random walk
from $w$ will take time $\Omega(k^2)$ to reach the horizontal line $H$ containing
$uv$, at which point its horizontal distance (along $H$) from $w$ will be distributed as a
$k^2$-step random horizontal random walk centered at $w$ (since about half of
the steps of the random walk up to that point will be horizontal). The difference in
probabilities between any two neighboring points in $H$ will therefore be at
most $O(1/k^2)$, which implies the bound on $|\phi(x)|$. Consequently, the
potential difference across any edge $wx$ at distance $k$ is at most $O(1/k^2)$;
since there are $O(k)$ edges at distance $k$, the total contribution from such
edges is $O(1/k)$, and summing over all distances $k$ (and repeating the argument for vertical edges) yields a bound of $O(\log
n)$.
\end{example}

\subsection{Our Results}
Our first theorem is the following.
\begin{theorem} \label{th.short} If $G=(V,E)$ is an unweighted graph with $m$ edges, then
$$\sum_{uv\in E} \Delta(u,v) \le O(m\log^2 n).$$
\end{theorem}
This theorem formalizes the intuition in the parallel paths example that there
cannot be too many edges in a graph for which the electrical flow uses very long
paths.  A corollary for edge-transitive graphs is that the above bound holds for
{\em every} edge, by symmetry. This generalizes our analysis on the grid (which
used very specific properties) to a much broader category which includes all
Cayley graphs.

Theorem \ref{th.short} is a consequence of a more general result concerning the
weighted transfer impedance matrix of a graph. Given a weighted graph $G=(V,E,c)$ with
edge weights $c_e\ge 0$, let $C$ be an $m\times m$ diagonal matrix containing
the edge weights. Then $L=B^TCB$ is the Laplacian matrix of $G$ and the {\em
weighted transfer impedance matrix} is the $m\times m$ matrix defined as: $$\Pi =
C^{1/2}BL^+B^TC^{1/2}.$$
It is well-known and easy to see that the entry $(BL^+B)(e,f)$ is the potential
difference across the ends of edge $e$ when a unit current is injected across
the ends of edge $f$, and vice versa, and that $\Pi$ is a projection matrix with trace
$n-1$. In particular, the latter fact implies that $\|\Pi\|=1$, where
$\|\cdot\|$ is the spectral norm. 

Let $\hP$ be the entrywise absolute value matrix of $\Pi$. Our main theorem is:
\begin{theorem} \label{th.absnorm} For an arbitrary weighted graph $G$,
$$\|\hP\|=O(\log^2 n)$$
\end{theorem}

Theorem \ref{th.short} follows from Theorem \ref{th.absnorm} by plugging in the
the all ones vector $u=(1,\ldots,1)^T$:
$$ u^T\hP u = \sum_{e\in E} \|\Pi_e\|_1 = \sum_{e\in E} \Delta(e),$$
where $\Pi_{e}=BL^+b_{uv}$ is the row of $\Pi$ corresponding to $e=uv$,
i.e., the electrical flow across the endpoints of $e$. Since $\|u\|^2=m$, the
spectral norm bound in Theorem \ref{th.absnorm} implies that $u^T\hP u\le
O(m\log^2 n)$.

\subsection{Applications to Oblivious Routing} Oblivious routing refers to
the following problem: given a graph $G$, specify a set of flows $\{f_{uv}\}$
between pairs of vertices $u,v$ so that for any set of demand pairs
$(s_1,t_1),\ldots,(s_k,t_k)$, the congestion of the flow obtained by taking the
union of $\{f_{s_it_i}\}_{i\le k}$ is at most a small factor (called the
competitive ratio) greater than the congestion of the optimal multicommodity
flow for the given pairs. This is a well-studied problem with a vast literature
which we will not attempt to recount; a landmark result is the optimal theorem
of R\"acke \cite{racke}
which shows that there is an oblivious routing scheme with competitive ratio
$O(\log n)$ for every graph. 

In spite of this optimal result, there has been interest in studying whether
simpler schemes achieve good competitive ratios. A particulary simple scheme,
studied in \cite{harsha08, lawler, kelner2011electric}, is to simply
route $f_{uv}$ using the electrical flow. The paper \cite{harsha08} shows that
this scheme has a good competitive ratio on any graph when restricted to demands
which all share a single source. It was shown in \cite{lawler,
kelner2011electric} that the
competitive ratio of electrical routing on an unweighted graph is exactly equal to $\|\Pi\|_{1\to 1}$, i.e.,
the maximum of $\Delta(u,v)$ over all edges in a graph. In these papers, it was
shown that for grids, hypercubes, and expanders the competitive ratio is $O(\log
n)$. Our theorem immediately extends this to all transitive graphs, albeit with
a guarantee of $O(\mathrm{polylog}(n))$ rather than $O(\log(n))$. 

\begin{corollary} Electrical Flow Routing achieves a competitive ratio of
$O(\log^2 n)$ on every edge-transitive graph.\end{corollary}
\begin{proof} By Theorem \ref{th.short} and symmetry, we have that every column
sum of $\hP$ is at most $O(\log^2 n)$. By Proposition 1 and Lemma 4 of
\cite{lawler} (or by Theorem 3.1 of \cite{kelner2011electric}), this implies that routing each pair by the electrical flow has a
competitive ratio of $O(\log^2 n)$ as an oblivious routing scheme. \end{proof}

\subsection{Techniques}
Given the expander example above, it may be tempting to attempt to prove Theorem
\ref{th.short} by decomposing an arbitrary graph into disjoint expanding
clusters. However, using such a decomposition would likely require proving that
edge electrical flows do not cross between the clusters, which is what we are
trying to show in the first place.

We use an alternate scheme inspired by recent
Schur-complement based Laplacian solvers. Recall the Schur complement formula
for the pseudoinverse of a symmetric block matrix (see e.g. \cite{durfee} Section 5):

\begin{fact}\label{fact:blockInverse}
If 
$$L = \begin{bmatrix} P & Q \\ Q^T & R\end{bmatrix}$$
for symmetric $P,R$ and $R$ invertible, then:
\begin{equation}\label{eqn:blockInverse}
L^{+}= 
Z^T
\begin{bmatrix}
I & 0 \\ -R^{-1}Q^T & I
\end{bmatrix}
\begin{bmatrix}
\schur{L}{P}^+ & 0 \\ 0 & R^{-1}
\end{bmatrix}
\begin{bmatrix}
I & -QR^{-1} \\ 0 & I
\end{bmatrix}
Z
\end{equation}
where $\schur{L}{P}=P-QR^{-1}Q^T$ denotes the Schur complement of $L$ onto $P$, obtained by
eliminating $R$ by partial Gaussian elimination, and $Z$ is the projection
orthogonal to the nullspace of $L$.
\end{fact}
The idea is to apply this formula to compute the terms
$|b_e^TL^+b_f|$ by eliminating vertices one-by-one, as in \cite{kyng}, and
bounding the original value of $|b_e^TL^+b_f|$ in terms of the value on
small Schur complements. One cannot eliminate arbitrary vertices and get
a good bound, though. We use Proposition \ref{prop:local-energy} to show
that there always exists a vertex whose elimination results in a good bound. Since Laplacian matrices with self loops are closed under taking Schur complements the remaining matrix is the Laplacian of a weighted graph as well. Mapping the demand vectors $b_e$ and $b_f$ to the vertex set of this graph and recurring yields the sum of interest.

\section{Schur Complements, Probabilities, and Energies}
In this section we collect some preliminary facts about Schur complements of
Laplacians and establish some useful correspondences between electrical
potentials and probabilities. We do this so that after recurring on a
Schur complement of the graph $G$ that we care about, we can interpret the recursively generated sums that we generate using Fact \ref{fact:blockInverse} in terms of $G$. We will make frequent use of the fact that for a Laplacian matrix $L_G$ with block $L_S$, the Schur complement
$\schur{L_G}{L_S}$ is also a Laplacian. For a graph $G$ and subset of vertices $S$ will use the notation $\schur{G}{S}$ to denote the graph  corresponding to $\schur{L_G}{L_S}$. Since all vectors we will apply pseudoinverses to will be orthogonal to the corresponding nullspaces (the corresponding constant vectors, since all Schur complements will be Laplacians), we will not write the projection $Z$
in Fact \ref{fact:blockInverse} in what follows.
\begin{definition}
Consider a graph $G$. For any set of vertices $S\subseteq V(G)$ with $|S|\ge 2$, a vertex $v\in S$, and a vertex $x\in V(G)$, let
$$p_v^{G,S}(x) := \P_x[t_v < t_{S\backslash v}]$$
For an edge $e = \{x,y\}\in E(G)$, let 
$$q_v^{G,S}(e) := |p_v^S(x) - p_v^S(y)|$$
where $t_{S'}$ denotes the hitting time to the set $S'$. Let
$$r_v^{G,S}(e) := \max(p_v^S(x),p_v^S(y),1/|S|)$$
When $G$ is clear from the context, we omit $G$ from the superscript.
\end{definition}

\begin{corollary}\label{prop:sum-potentials}
For any set $S$ and any $e = \{x,y\}\in E(G)$, $\sum_{v\in S} r_v^{G,S}(x) \le 3$.
\end{corollary}
\begin{proof}
$\{p_v^{G,S}(z)\}_v$ is a distribution for any fixed vertex $z\in V(G)$. Bounding $r_v^{G,S}(e)\le p_v^{G,S}(x) + p_v^{G,S}(y) + \frac{1}{|S|}$ yields the desired result.
\end{proof}

It is well-known that the above probabilities can be represented as normalized
potentials (see, for instance, \cite{bollobas2013modern} Chapter IX).
\begin{fact}\label{rmk:elec-prob-link}
Let $H$ be the graph obtained by identifying all vertices of $S\backslash \{v\}$ to one vertex $s$. Then $p_v^S(x) := \frac{|b_{vs}^T L_H^+ b_{xs}|}{b_{vs}^T L_H^+ b_{vs}}$ for any $x\in V(G)$ and $q_v^S(e) := \frac{|b_{vs}^T L_H^+ b_e|}{b_{vs}^T L_H^+ b_{vs}}$ for any $e\in E(G)$.
\end{fact}

In proving the desired result, it will help to recursively compute Schur
complements with respect to certain sets of vertices $S$. We now relate the
Schur complement to the above probabilities, which will be central to our proof;
the following proposition is likely to be known but we include it for completeness.

\begin{proposition}\label{prop:schur-probabilities}
For a set of vertices $S\subseteq V(G)$ and a vertex $x\in V(G)$ possibly not in $S$, let $b_x\in \mathbb{R}^{V(G)}$ denote the indicator vector of $x$. Let $b_x^S\in \mathbb{R}^{S}$ denote the vector with coordinates $b_x^S(v) = p_v^S(x)$ for all $v\in S$. Write $L_G$ as a two-by-two block matrix:

\[
L=
  \begin{bmatrix}
    P & Q \\
    Q^T & R
  \end{bmatrix}
\]
where $P$, $Q$, and $R$ have index pairs $S\times S$, $S\times (V(G)\setminus S)$, and $V(G)\setminus S\times V(G)\setminus S$ respectively. 
Then
$$b_x^S = M_S b_x$$
where 
$$M_S=\begin{bmatrix} I & -QR^{-1}\end{bmatrix}.$$\end{proposition}

\begin{proof}
If $x\in S$, then $b_x^S$ is the indicator vector of $x$ and $x$ is in the identity block of $M_S$. Therefore, $b_x^S = M_S b_x$.

If $x\notin S$, then let $b_x^c$ denote the coordinate restriction of $b_x$ to $V(G)\setminus S$. We want to show that $b_x^S = -QR^{-1}b_x^c$. Consider the linear system

$$b_x^c = Rp$$

Let $H$ be the graph obtained by identifying all vertices in $S$ within $G$ to a
single vertex $s$. Then the vector $p'$ with $p_s' = 0$ and $p_v' = p_v$ for all
$v\in V(H)\setminus \{s\}$ is a solution to a boundary value problem with $p_s'
= 0$ and $p_x'$ having the maximum potential of any vertex. The block matrix $Q$
can be viewed as mapping the potentials $p'$ to a flow proportional to the
$xs$-flows on edges incident with $s$. By Proposition 2.2 of \cite{LP16}, for
example, the incoming flow on edges to $s$ is equal to the probability that an
$x\rightarrow s$ random walk first visits $s$ by crossing that edge. Grouping
edges according to their common endpoints shows that $-QR^{-1}b_x^c$ is a scalar
multiple of $b_x^S$.

However, notice that

$$\textbf{1}^T (-QR^{-1})b_x^c = (-\textbf{1}^T Q)R^{-1}b_x^c = \textbf{1}^T R R^{-1} b_x^c = 1 = \textbf{1}^T b_x^S$$

so $b_x^S = -QR^{-1}b_x^c$, as desired.
\end{proof}

Once one views the the $q_v^S(e)$s in the above way, it makes sense to discuss
the energy of the $q_v^S(e)$s in relation to the probabilities $p_v^S(x)$. It
turns out that the total energy contributed by edges with both endpoints having
potential at most $p$ is at most a $p$ fraction of the total energy.

\begin{proposition}\label{prop:norm-energy}
For any $p\in (0,1)$, let $F$ be the set of edges $\{x,y\}$ with $\max_{z\in \{x,y\}} p_v^S(z)\le p$. Then the total energy of those edges is at most a $p$ fraction of the overall energy. More formally,

$$\sum_{e \in F} c_e (q_v^S(e))^2 \le p\sum_{e\in E(G)} c_e (q_v^S(e))^2$$
\end{proposition}

\begin{proof}
Let $H$ be the graph obtained by identifying $S\backslash \{v\}$ to a vertex $s$
in $G$. By Fact \ref{rmk:elec-prob-link}, we can show the desired proposition by proving the following:

$$\sum_{e\in F} c_e (b_{vs}^T L_H^+ b_e)^2 \le p b_{vs}^T L_H^+ b_{vs}$$

for an arbitrary graph $H$ and the subset of edges $F\subseteq E(H)$ with $\max_{z\in \{x,y\}} |b_{vs}^T L_H^+ b_{zs}| \le p b_{vs}^T L_H^+ b_{vs}$.

Write the sum on the left side in terms of an integral over sweep cuts. For $p\in (0,1)$, let $C_p$ denote the set of edges cut by the normalized potential $p$. More precisely, let $C_p$ be the set of edges $\{x,y\}$ with $|b_{vs}^T L_H^+ b_{xs}|\ge p b_{vs}^T L_H^+ b_{vs}$ and $|b_{vs}^T L_H^+ b_{ys}|\le p b_{vs}^T L_H^+ b_{vs}$. Notice that

\begin{align*}
\sum_{e\in F} c_e (b_{vs}^T L_H^+ b_e)^2 &\le \int_0^{p b_{vs}^T L_H^+ b_{vs}} \sum_{e\in C_q} c_e |b_{vs}^T L_H^+ b_e| dq\\
&= \int_0^{p b_{vs}^T L_H^+ b_{vs}} 1 dq\\
&\le p b_{vs}^T L_H^+ b_{vs}\\
\end{align*}

where the equality follows from the fact that $C_q$ is a threshold cut for the $v-s$ electrical flow and the first inequality follows from splitting the contribution of $e$ to the sum in terms of threshold cuts. This inequality is the desired result.
\end{proof}

Finally, we relate the weighted degrees of of vertices  in $\schur{G}{S}$ to energies in $G$ with respect to $S$.
\begin{definition} Let $c_v^H$ denote the sum of the conductances\footnote{To
avoid confusion, we remind the reader that by conductances we always mean electrical
conductances, i.e., weights in the graph, and not conductances in the sense of
expansion.}of edges
incident with $v$ in $H$.\end{definition}
\begin{proposition}\label{prop:schur-conductances}
Let $G$ be a graph. Consider a set $S$ and let $H = \texttt{Schur}(G,S)$. Then

$$c_v^H = \sum_{e\in E(G)} c_e^G (q_v^{G,S}(e))^2$$
\end{proposition}

\begin{proof}
Let $I$ be the graph obtained by identifying $S\setminus \{v\}$ to $s$ in $H$. Since the effective conductance of parallel edges is the sum of the conductance of those edges, $c_v^H = \frac{1}{b_{vs}^T L_I^+ b_{vs}}$.

By commutativity of Schur complements, $I$ can also be obtained by identifying $S\backslash \{v\}$ in $G$ before eliminating all vertices besides $s$ and $v$. Let $J$ be the graph obtained by just doing the first step (identifying $S\backslash \{v\}$). By definition of Schur complements,

$$b_{vs}^T L_I^+ b_{vs} = b_{vs}^T L_J^+ b_{vs}$$

By Fact \ref{rmk:elec-prob-link},

$$b_{vs}^T L_J^+ b_{vs} = \frac{1}{\sum_{e\in E(G)} c_e^G (q_v^{G,S}(e))^2}$$

Substitution therefore shows the desired result.
\end{proof}

\section{Proof of Theorem \ref{th.absnorm}}
We will deduce Theorem \ref{th.absnorm} from the following seemingly weaker
statement regarding positive test vectors.
\begin{theorem}\label{thm:nikhil}
Let $G$ be a graph. Then for any vector $w\in \mathbb{R}^{E(G)}_{\ge 0}$,

$$\sum_{e, f\in E} w_{e} w_{f}\sqrt{c_ec_f}|b_e^T L_G^+ b_f|\le O(\log^2 n) ||w||_2^2$$
\end{theorem}

Theorem \ref{th.absnorm} can be deduced from this by a Perron-Frobenius
argument. \begin{proof}[Proof of Theorem \ref{th.absnorm}]
Since the matrix $M = |C_G^{1/2}B_GL_G^+B_G^TC_G^{1/2}|$ has nonnegative entries, there is an eigenvector with maximum eigenvalue with nonnegative coordinates by Perron-Frobenius. Such an eigenvector corresponds to a positive eigenvalue. Theorem \ref{thm:nikhil} bounds the value of the quadratic form of this eigenvector. In particular, the quadratic form is at most $O(\log^2 n)$ times the $\ell_2$ norm squared of the vector, as desired.
\end{proof}

The proof hinges on the following key quantity. Define

$$\texttt{Degree}_S(u) := \frac{(\sum_{e\in E(G)} w_e\sqrt{c_e} q_u^S(e))^2}{\sum_{e\in E(G)} c_e q_u^S(e)^2}$$

The quantity $\texttt{Degree}_S(u)$ may be interpreted as a measure of the sparsity of the vector $(q_u^S(e))_e$, since it is the ratio of the (weighted) $\ell_1^2$ norm of this vector and its $\ell_2^2$ norm. Note that when $S=V(G)$, $w = \textbf{1}$, and $G$ is unweighted, $\texttt{Degree}_S(u)$ is simply the degree of $u$.

There are two parts to the proof: (1) recursively reducing the original problem to a number of problems involving sums of simpler inner products and (2) bounding those sums. The difference between the value of a problem and the subproblem after eliminating $u$ is at most $\texttt{Degree}_S(u)$. We want to show that there always is a choice of $u$ with a small value of $\texttt{Degree}_S(u)$. The following proposition shows this:

\begin{proposition}\label{prop:local-energy}
For any $\{c_e\}_e$-weighted graph $G$, set $S\subseteq V(G)$ with $|S|\ge 2$, and nonnegative weights $\{w_e\}_{e\in E(G)}$, the following holds:

$$\sum_{u\in S} \texttt{Degree}_S(u)\le O(\log |S|) \sum_{e\in E(G)} w_e^2$$
\end{proposition}

We now reduce Theorem \ref{thm:nikhil} to this proposition by picking the vertex $u$ with that minimizes the summand $\texttt{Degree}_S(u)$ of Proposition \ref{prop:local-energy} and recurring on the Schur complement with $u$ eliminated. The summand
of Proposition \ref{prop:local-energy} is an upper bound on the decrease
due to eliminating $u$.

\begin{proof}[Proof of Theorem \ref{thm:nikhil} given Proposition \ref{prop:local-energy}]
Define the following:

\begin{itemize}
\item $G_0 \gets G$, $c^{(0)}\gets c$, $x_0 \gets \arg\min_{x\in V(G)}\texttt{Degree}_{S_0}(x)$, $S_0 \gets V(G)$, $i\gets 0$.
\item While $|V(G_i)| > 2$:
    \begin{itemize}
    \item $i\gets i+1$
    \item $G_i \gets \texttt{Schur}(G_{i-1}, V(G_{i-1})\setminus \{x_{i-1}\})$
    \item $c^{(i)} \gets $ conductance vector for $G_i$
    \item $S_i \gets V(G_i)$
    \item $x_i \gets \arg\min_{x\in V(G_i)} \texttt{Degree}_{S_i}(x)$
    \end{itemize}
\item $T\gets i$
\end{itemize}

Let $L_i \gets L_{G_i}$ and let $m_i = L_{x_ix_i}$. We start by understanding how to
express the left hand side of the desired inequality in $G_i$ for all $i$. For a
vertex $x\in V(G)$, let $b_x^{(i)}\in \mathbb{R}^{V(G_i)}$ denote the vector
with $b_x^{(i)}(v) = p_v^{G,S_i}(x)$ for all $v\in V(G_i)$. For an edge
$\{x,y\}\in V(G)$, let $b_{xy}^{(i)} = b_x^{(i)} - b_y^{(i)}$. Let

$$\mc{V}_i := \sum_{e,f\in E(G)} w_e \sqrt{c_e} |b_e^{(i)T} L_i^+ b_f^{(i)}| \sqrt{c_f} w_f$$

We now bound $\mc{V}_i$ in terms of $\mc{V}_{i+1}$ for all nonnegative integers
$i < T$. By Fact \ref{fact:blockInverse} and Proposition \ref{prop:schur-probabilities},
\begin{align*}
\mc{V}_i &= \sum_{e,f\in E(G)} w_e \sqrt{c_e} |b_e^{(i)T} L_i^+ b_f^{(i)}| \sqrt{c_f} w_f\\
&= \sum_{e,f\in E(G)} w_e \sqrt{c_e} |b_e^{(i+1)T} L_{i+1}^+ b_f^{(i+1)} + x_e^{(i)T} \frac{1}{m_i} x_f^{(i)}| \sqrt{c_f} w_f\\
&\le \mc{V}_{i+1} + \sum_{e,f\in E(G)} w_e \sqrt{c_e} |x_e^{(i)T} \frac{1}{m_i} x_f^{(i)}| \sqrt{c_f} w_f\\
\end{align*}
where $x_e^{(i)} := b_e^{(i)}(x_i)$. Since the $x_e^{(i)}$s are scalars,
we can futher simplify the above sum:

\begin{align*}
\sum_{e,f\in E(G)} w_e \sqrt{c_e} |x_e^{(i)T} \frac{1}{m_i} x_f^{(i)}| \sqrt{c_f} w_f &= \frac{1}{m_i}\left(\sum_{e\in E(G)} w_e \sqrt{c_e} |x_e^{(i)}|\right)^2\\
&= \frac{1}{c_{x_i}^{(i)}}\left(\sum_{e\in E(G)} w_e \sqrt{c_e} x_e^{(i)}\right)^2\\
&= \frac{\left(\sum_{e\in E(G)} w_e \sqrt{c_e} q_{x_i}^{S_i}(e)\right)^2}{\sum_{e\in E(G)} c_eq_{x_i}^{S_i}(e)^2}\\
&= \texttt{Degree}_{S_i}(x_i)\\
\end{align*}

where the second-to-last denominator equality follows from Proposition \ref{prop:schur-conductances}. Since $x_i$ minimizes $\texttt{Degree}_{S_i}(x_i)$, Proposition \ref{prop:local-energy} with $S\gets S_i$ and $G\gets G$ implies that

\begin{align*}
\texttt{Degree}_{S_i}(x_i)&\le O\left(\frac{\log n}{|S_i|}||w||_2^2\right)\\
&\le O\left(\frac{\log n}{n-i}||w||_2^2\right)\\
\end{align*}

Plugging this in shows that

$$\mc{V}_i\le \mc{V}_{i+1} + O\left(\frac{\log n}{n-i}\right) ||w||_2^2$$
for all $i < T$. Therefore, to bound $\mc{V}_0$, it suffices to bound $\mc{V}_T$. Let $S_{T} = \{a,b\}$. Then

\begin{align*}
\mc{V}_{T} &= \sum_{e,f\in E(G)} w_e \sqrt{c_e} |b_e^{(T)T}L_{T}^+b_f^{(T)}| \sqrt{c_f}w_f\\
&= \sum_{e,f\in E(G)} w_e \sqrt{c_e} q_{u_i}^{G,S_i}(e)|b_{ab}^TL_{T}^+b_{ab}|q_{u_i}^{G,S_i}(f) \sqrt{c_f}w_f\\
&= \sum_{e,f\in E(G)} w_e \sqrt{c_e} \frac{|b_{ab}^T L_{T}^+ b_e|}{b_{ab}^T L_{T}^+ b_{ab}}b_{ab}^TL_{T'}^+b_{ab}\frac{|b_{ab}^T L_{T}^+ b_f|}{b_{ab}^T L_{T}^+ b_{ab}} \sqrt{c_f}w_f\\
&\le ||w||_2^2
\end{align*}

where the last line follows from Cauchy-Schwarz. Therefore,

$$\mc{V}_{T}\le ||w||_2^2$$

Combining these bounds yields a harmonic sum that proves the desired result.
\end{proof}

Now, we prove Proposition \ref{prop:local-energy}.

\begin{proof}
For each vertex $v\in S$ and each integer $i\in [0,\log |S|]$, let $X_v^{(i)}\subseteq E(G)$ denote the set of edges $e = \{x,y\}$ for which $r_v^S(e) \le 2^{-i}$. Let $T := \log |S|$. For each $0\le i < T$, let $Y_v^{(i)} = X_v^{(i)}\setminus X_v^{(i+1)}$. Let $Y_v^{(T)} = X_v^{(T)}$.

For each $v$ and each $i\ge 0$, $X_v^{(0)} = E(G)$, and $X_v^{(i+1)}\subseteq X_v^{(i)}$. Therefore, $\{Y_v^{(i)}\}_{i=0}^T$ is a partition of $E(G)$ for each $v\in S$. By Cauchy-Schwarz,

\begin{align*}
\sum_{u\in S} \texttt{Degree}_S(u) &= \sum_{u\in S} \frac{(\sum_{e\in E(G)} w_e\sqrt{c_e} q_u^S(e))^2}{\sum_{e\in E(G)} c_e q_u^S(e)^2}\\
&\le \sum_{u\in S}\left(\sum_{e\in E(G)} r_u^S(e)w_e^2\right)\frac{\sum_{e\in E(G)} c_e q_u^S(e)^2/r_u^S(e)}{\sum_{e\in E(G)} c_e q_u^S(e)^2}\\
\end{align*}

By the definition of $X_u^{(i+1)}$,

$$\sum_{u\in S}\left(\sum_{e\in E(G)} r_u^S(e)w_e^2\right)\frac{\sum_{e\in E(G)} c_e q_u^S(e)^2/r_u^S(e)}{\sum_{e\in E(G)} c_e q_u^S(e)^2} \le \sum_{u\in S}\left(\sum_{e\in E(G)} r_u^S(e)w_e^2\right)\left(\sum_{i=0}^T 2^{i+1}\frac{\sum_{e\in Y_u^{(i)}} c_e q_u^S(e)^2}{\sum_{e\in E(G)} c_e q_u^S(e)^2}\right)$$

By the definition of $X_u^{(i)}$ and Proposition \ref{prop:norm-energy},

\begin{align*}
\sum_{u\in S}\left(\sum_{e\in E(G)} r_u^S(e)w_e^2\right)\left(\sum_{i=0}^T 2^{i+1}\frac{\sum_{e\in Y_u^{(i)}} c_e q_u^S(e)^2}{\sum_{e\in E(G)} c_e q_u^S(e)^2}\right)&\le \sum_{u\in S}\left(\sum_{e\in E(G)} r_u^S(e)w_e^2\right)\left(\sum_{i=0}^T 2\right)\\
&\le \sum_{u\in S}\left(\sum_{e\in E(G)} r_u^S(e)w_e^2\right)\left(2T + 2\right)\\
\end{align*}

By Proposition \ref{prop:sum-potentials},

$$\sum_{u\in S}\left(\sum_{e\in E(G)} r_u^S(e)w_e^2\right)\left(2T + 2\right)\le (6T + 6) \sum_{e\in E(G)} w_e^2$$

Combining these bounds shows that

$$\sum_{u\in S} \texttt{Degree}_S(u)\le (6T + 6) \sum_{e\in E(G)} w_e^2 \le O(\log |S|) ||w||_2^2$$

as desired.
\end{proof}

\subsection*{Acknowledgments} We would like to thank Hariharan Narayanan, Akshay Ramachandran, and Hong Zhou for many helpful conversations.

\bibliography{localization}

\newcommand{\etalchar}[1]{$^{#1}$}
\begin{thebibliography}{CKM{\etalchar{+}}11}

\bibitem[Bol13]{bollobas2013modern}
B{\'e}la Bollob{\'a}s.
\newblock {\em Modern graph theory}, volume 184.
\newblock Springer Science \& Business Media, 2013.

\bibitem[CKM{\etalchar{+}}11]{CKMST}
Paul Christiano, Jonathan~A Kelner, Aleksander Madry, Daniel~A Spielman, and
  Shang-Hua Teng.
\newblock Electrical flows, laplacian systems, and faster approximation of
  maximum flow in undirected graphs.
\newblock In {\em Proceedings of the forty-third annual ACM symposium on Theory
  of computing}, pages 273--282. ACM, 2011.

\bibitem[DKP{\etalchar{+}}17]{durfee}
David Durfee, Rasmus Kyng, John Peebles, Anup~B Rao, and Sushant Sachdeva.
\newblock Sampling random spanning trees faster than matrix multiplication.
\newblock In {\em Proceedings of the 49th Annual ACM SIGACT Symposium on Theory
  of Computing}, pages 730--742. ACM, 2017.

\bibitem[HHN{\etalchar{+}}08]{harsha08}
Prahladh Harsha, Thomas~P Hayes, Hariharan Narayanan, Harald R{\"a}cke, and
  Jaikumar Radhakrishnan.
\newblock Minimizing average latency in oblivious routing.
\newblock In {\em Proceedings of the nineteenth annual ACM-SIAM symposium on
  Discrete algorithms}, pages 200--207. Society for Industrial and Applied
  Mathematics, 2008.

\bibitem[KM09]{kelner2009faster}
Jonathan~A Kelner and Aleksander Madry.
\newblock Faster generation of random spanning trees.
\newblock In {\em Foundations of Computer Science, 2009. FOCS'09. 50th Annual
  IEEE Symposium on}, pages 13--21. IEEE, 2009.

\bibitem[KM11]{kelner2011electric}
Jonathan Kelner and Petar Maymounkov.
\newblock Electric routing and concurrent flow cutting.
\newblock {\em Theoretical Computer Science}, 412(32):4123--4135, 2011.

\bibitem[KMP12]{kelner2012faster}
Jonathan~A Kelner, Gary~L Miller, and Richard Peng.
\newblock Faster approximate multicommodity flow using quadratically coupled
  flows.
\newblock In {\em Proceedings of the forty-fourth annual ACM symposium on
  Theory of computing}, pages 1--18. ACM, 2012.

\bibitem[KS16]{kyng}
Rasmus Kyng and Sushant Sachdeva.
\newblock Approximate gaussian elimination for laplacians: Fast, sparse, and
  simple.
\newblock {\em CoRR}, abs/1605.02353, 2016.

\bibitem[LN09]{lawler}
Gregory Lawler and Hariharan Narayanan.
\newblock Mixing times and l p bounds for oblivious routing.
\newblock In {\em Proceedings of the Meeting on Analytic Algorithmics and
  Combinatorics}, pages 66--74. Society for Industrial and Applied Mathematics,
  2009.

\bibitem[LP16]{LP16}
Russell Lyons and Yuval Peres.
\newblock {\em Probability on Trees and Networks}, volume~42 of {\em Cambridge
  Series in Statistical and Probabilistic Mathematics}.
\newblock Cambridge University Press, New York, 2016.

\bibitem[LRS13]{LRS}
Yin~Tat Lee, Satish Rao, and Nikhil Srivastava.
\newblock A new approach to computing maximum flows using electrical flows.
\newblock In {\em Proceedings of the forty-fifth annual ACM symposium on Theory
  of computing}, pages 755--764. ACM, 2013.

\bibitem[Mad13]{madry13}
Aleksander Madry.
\newblock Navigating central path with electrical flows: From flows to
  matchings, and back.
\newblock In {\em Foundations of Computer Science (FOCS), 2013 IEEE 54th Annual
  Symposium on}, pages 253--262. IEEE, 2013.

\bibitem[MST15]{MST}
Aleksander Madry, Damian Straszak, and Jakub Tarnawski.
\newblock Fast generation of random spanning trees and the effective resistance
  metric.
\newblock In {\em Proceedings of the Twenty-Sixth Annual ACM-SIAM Symposium on
  Discrete Algorithms}, pages 2019--2036. Society for Industrial and Applied
  Mathematics, 2015.

\bibitem[Rac02]{racke}
Harald Racke.
\newblock Minimizing congestion in general networks.
\newblock In {\em Foundations of Computer Science, 2002. Proceedings. The 43rd
  Annual IEEE Symposium on}, pages 43--52. IEEE, 2002.

\bibitem[SS11]{spielman2011graph}
Daniel~A Spielman and Nikhil Srivastava.
\newblock Graph sparsification by effective resistances.
\newblock {\em SIAM Journal on Computing}, 40(6):1913--1926, 2011.

\bibitem[ST04]{spielmanteng}
Daniel~A Spielman and Shang-Hua Teng.
\newblock Nearly-linear time algorithms for graph partitioning, graph
  sparsification, and solving linear systems.
\newblock In {\em Proceedings of the thirty-sixth annual ACM symposium on
  Theory of computing}, pages 81--90. ACM, 2004.

\end{thebibliography}

\end{document}